\documentclass[11pt]{article} %{llncs} %,a4paper

\usepackage[utf8x]{inputenc}
\usepackage{amsthm}
\usepackage{amsmath}
\usepackage{amsfonts}
\usepackage{amssymb}
\usepackage{paralist}
\usepackage{graphicx}
\usepackage{a4wide}
\usepackage[margin=2.8cm,nohead]{geometry}
 
\usepackage{comment}
\usepackage{color}
\usepackage{verbatim}
\usepackage{xspace}
\usepackage{algorithmicx}
\usepackage{algorithm}% http://ctan.org/pkg/algorithm
\usepackage{xspace}
\usepackage{wrapfig}
\usepackage{enumitem}
\usepackage[export]{adjustbox}
\usepackage[noend]{algpseudocode}% http://ctan.org/pkg/algorithmicx
\usepackage[blocks]{authblk}
\usepackage[unicode]{hyperref}

\def\C{{\mathcal{C}}}
\def\calD{{\mathcal{D}}}
\def\calE{{\mathcal{E}}}

\def\O{{\mathcal{O}}}

\def\S{{\mathcal{S}}}
\let\epsilon=\varepsilon
\def\I{\it\aftergroup\/}

\def\text#1{\hbox{#1}}
\let\t=\mbox
\def\set{\mbox{set}}
\def\unset{\mbox{unset}}
\def\false{\mbox{false}}
\def\D{{\sc Drawer}\xspace}
\def\P{{\sc Painter}\xspace}
\def\GR{{\sc Greedy}\xspace}
\def\WIN{{\sc Winning}\xspace}
\def\VIRTINIT{{\sc InitSimulation}\xspace}
\def\VIRT{{\sc ColorBySimulation}\xspace}
\def\WFD{{\sc WaitFor}$D$\xspace}
\def\FF{{\sc FirstFit}\xspace}
\def\uvertex{$\forall$-vertex\xspace}
\def\evertex{$\exists$-vertex\xspace}
\def\uvertices{$\forall$-vertices\xspace}
\def\evertices{$\exists$-vertices\xspace}

\def\enumparams{\itemsep1pt \parskip0pt}
\def\indentskip{\hskip 1.5em}

\algnewcommand{\IIf}[1]{\State\algorithmicif\ #1\ \algorithmicthen}
\algnewcommand{\IElse}{\algorithmicelse\ }
\algnewcommand{\EndIIf}{\unskip\ \algorithmicend\ \algorithmicif}

 \PrerenderUnicode{č}\PrerenderUnicode{š}\PrerenderUnicode{ž}\PrerenderUnicode{ř}
 \PrerenderUnicode{í}\PrerenderUnicode{á}\PrerenderUnicode{ý}\PrerenderUnicode{ť}
 \PrerenderUnicode{ď}\PrerenderUnicode{ě}\PrerenderUnicode{é}\PrerenderUnicode{ö}

\newcommand{\onlinecol}{\textsc{Online Graph Coloring}\xspace}
\newcommand{\onlinechrom}{\textsc{Online Chromatic Number}\xspace}
\newcommand{\onlineprecoloring}{\textsc{Online Chromatic Number with Precoloring}\xspace}
\newcommand{\virt}[1]{#1_{\rm{virt}}}
\newcommand\boldsection[1]{\medskip \noindent\textbf{#1}}

\begin{document}

% Trochu volnejsi nastaveni deleni slov
\lefthyphenmin=2
\righthyphenmin=2

\newtheorem*{define}{Definition}
\newtheorem{theorem}{Theorem}[section]
\newtheorem{lemma}[theorem]{Lemma}
\newtheorem{claim}[theorem]{Claim}
\newtheorem{observation}[theorem]{Observation}
\newtheorem{corollary}[theorem]{Corollary}
\newtheorem{proposition}[theorem]{Proposition}

\theoremstyle{definition}
\newtheorem{definition}{Definition}

\long\def\algobox#1{\smallskip
  \noindent
~~\hbox{\fbox{\parbox[c]{0.95\textwidth}{\small #1}}}
\smallskip
}

%TODO add some other pictures?
%TODO better typesetting of algs.

\date{}
\title{Online Chromatic Number is PSPACE-Complete\thanks{
Supported by project 17-09142S of GA \v{C}R and by the GAUK project 634217.
A preliminary version of this work appeared in~\cite{BohVel16}.
}}
\author{Martin Böhm}
\author{Pavel Veselý}
\affil{Computer Science Institute of Charles University, Prague, Czech Republic.
\texttt{\{bohm,vesely\}@iuuk.mff.cuni.cz}.}

\maketitle

\begin{abstract}
In the online graph coloring problem, vertices from a graph $G$, known in advance, arrive
in an online fashion and an algorithm must immediately assign a color to each incoming vertex $v$
so that the revealed graph is properly colored. The exact location of $v$ in the graph $G$
is not known to the algorithm, since it sees only previously colored neighbors of $v$.
The \emph{online chromatic number} of $G$ is the smallest number
of colors such that some online algorithm is able to properly color $G$ for any incoming order.
We prove that computing the online chromatic number of a graph is PSPACE-complete.
\end{abstract}

\section{Introduction} \label{sec:intro}
In the classical graph coloring problem we assign a color to each vertex of
a given graph such that the graph is properly colored, i.e., no two adjacent vertices have the same color.
The chromatic number $\chi$ of a graph $G$ is the smallest $k$ such that
$G$ can be colored with $k$ distinct colors.
Deciding whether the chromatic number of a graph is at most $k$ is well known to be NP-complete,
even in the case with three colors. 

%In the setting of \emph{online problems}, the input (in our case,
%the graph $G$) is not received in advance, but arrives sequentially,
%and the task (coloring vertices) is performed immediately
%and irrevocably.

The online variant of graph coloring can be defined as follows: The
vertices of $G$ arrive one by one, and an online algorithm must color
vertices as they arrive so that the revealed graph is properly colored
at all times. When a vertex arrives, the algorithm sees edges to
previously colored vertices. The online algorithm may use additional knowledge of the whole graph $G$;
more precisely, a copy of $G$ is sent to the algorithm at the start of the input. However, the exact
correspondence between the incoming vertices and the vertices of the
copy of $G$ is not known to the algorithm. This problem is called \onlinecol.

In this paper we focus on a graph parameter called \emph{online chromatic number}
which is analogous to the standard chromatic number of a graph.

\begin{definition}
The \emph{online chromatic number} $\chi^O(G)$ of a graph $G$ is the
smallest number $k$ such that there exists a deterministic online
algorithm which is able to color the specified graph $G$ using $k$
colors for any incoming order of vertices.
\end{definition}

The online chromatic number has been studied since 1990
\cite{GryLeh88}. One of the main open problems in the area is the
computational complexity of deciding whether $\chi^O(G)\leq k$ for a
specified simple graph $G$, given $G$ and $k$ on input; see
e.g.\ Kudahl \cite{kudahl14}. %We denote this decision problem as \onlinechrom.

\begin{definition}
The \onlinechrom problem is as follows: 

\noindent \textbf{Input:} An undirected simple graph $G$ and an integer $k$.

\noindent \textbf{Goal:} Decide whether $\chi^O(G)\leq k$.
\end{definition}

In this paper, we fully resolve the computational complexity of this problem:

\begin{theorem}\label{thm:main}
The decision problem \onlinechrom is PSPACE-complete.
\end{theorem}

As is usual in the online computation model, we can view \onlinecol as a game
between two players, which we call \P (representing the online algorithm) and \D (often
called \textsc{Adversary} in the online algorithm literature). 
In each round \D chooses an uncolored vertex $v$ from $G$
and sends it to \P without without any information to which vertex of $G$
it corresponds, only revealing the edges to the previously sent vertices.
Then \P must properly color (``paint'') $v$,
i.e., \P cannot use a color of a neighbor of $v$.
We stress that in this paper \P is restricted to be deterministic.
The game continues with the next round until all vertices of $G$ are colored.

%Keep in mind that while \onlinecol is an online
%problem, \onlinechrom is an (offline) decision problem of checking
%whether $\chi^O(G) \le k$.

\boldsection{Examples.} Consider a path $P_4$ on four vertices. Initially, \D sends
two nonadjacent vertices. If \P assigns different colors to them, then
these are the first and the third vertex of $P_4$, thus the second vertex must
get a third color; otherwise they obtained the same color $a$ and they are the endpoints
of $P_4$, therefore the second and the third vertex get different colors which are not equal to $a$.
In both cases, there are three colors on $P_4$ and thus $\chi^O(P_4)=3$, while $\chi(P_4)=2$.

Note also that we may think of \D deciding where an incoming vertex
belongs at some time \textit{after} it is colored provided that \D's
choice still allows for at least one isomorphism to the original
$G$. This is possible only for a deterministic \P.

A particularly interesting class of graphs in terms of $\chi^O$ is the
class of binomial trees.  A binomial tree of order $k$, denoted $B_k$, is defined
inductively: $B_0$ is a single vertex (the
root) and $B_k$ is created by taking two
disjoint copies of $B_{k-1}$, adding an edge
between their roots and choosing one of their roots as the root for
the resulting tree.  Thus $P_4$ with the root on the second vertex is $B_2$.

It is not hard to generalize the example of $P_4$
and show $\chi^O(B_k) = k+1$~\cite{GryLeh88}.
This shows that the ratio between $\chi^O$ and $\chi$ may be
arbitrarily large even for the class of trees.

\boldsection{History.}
%The study of online problems is usually concerned with
%finding an efficient online algorithm with output that can be compared
%to an offline algorithm.
The online problem \onlinecol
has been known since 1976 \cite{bean76}, originally studied in the
variant where the algorithm has no extra information at the start of
the input. Bean~\cite{bean76} showed that no online algorithm that is compared
to an offline algorithm can perform well under this metric.
The notion of online chromatic number was first defined in
1990 by \cite{GryLeh88}.

For the online problem, Lovász, Saks and Trotter \cite{lst} show an
algorithm with a \emph{competitive ratio} $O(n / \log^* n)$, where the
competitive ratio is the ratio of the number of colors used by the
online algorithm to the (standard) chromatic number. This was later
improved to $O(n \log \log \log n / \log \log n)$ by Kierstad \cite{kierstad98}
using a deterministic algorithm. %TODO check the result by Kierstad
There is a better $O(n / \log n)$-competitive randomized algorithm against an
oblivious adversary by Halldórsson~\cite{halldorsson97}.
A lower bound on the competitive ratio of
$\Omega(n/\log^2n)$ even for randomized algorithms against an oblivious
adversary was shown by Halldórsson and Szegedy~\cite{hallszegedy94}.

Our variant of \onlinecol, where the algorithm receives a copy of the
graph at the start, was suggested by Halldórsson \cite{halldorsson00},
where it is shown that the lower bound $\Omega(n/\log^2n)$ also holds
in this model. (Note that the previously mentioned algorithmic results
are valid for this model also.)

Kudahl~\cite{kudahl13thesis} recently studied \onlinechrom as a
complexity problem. The paper shows that the problem is coNP-hard and
lies in PSPACE.  Later~\cite{kudahl14} he proved that if some part of
the graph is precolored, i.e., some vertices are assigned some colors
prior to the coloring game between \D and \P and \D also reveals edges
to the precolored vertices for each incoming vertex, then deciding
whether $\chi^O(G) \leq k$ is PSPACE-complete.  We call this decision
problem \onlineprecoloring.  The paper~\cite{kudahl14} conjectures
that \onlinechrom (with no precolored part) is PSPACE-complete
too. Interestingly, it is possible to decide $\chi^O(G) \leq 3$ in
polynomial time~\cite{GryKirLeh93}.

\boldsection{Related work.} Deciding the outcome of many two-player games is PSPACE-complete;
among those are (generalizations of) Amazons, Checkers and Hex, to name a few.

The closest PSPACE-complete two-player game to our model is arguably
\textsc{Sequential Coloring Game}, where two players color vertices in
a fixed order and the first player to use more than $k$ colors loses
the game. This game was defined and analyzed by
Bodlaender~\cite{Bodlaender91}.

However, in all of the games mentioned above, both players have roughly
the same power. % (except that one starts the game).
This does not hold for \onlinecol which is highly asymmetric, since \D has perfect
information (knows which vertices are sent and how they are colored),
but \P does not. \P may only guess to which part of the graph does the
colored subgraph really belong.  This is the main difficulty in
proving PSPACE-hardness.

An example of an asymmetric two-player game somewhat related to
\onlinecol is the \textsc{Online Ramsey Game} in which Builder draws
edges and Painter colors each edge either red, or blue.  Builder wins
if it forces a monochromatic copy of a graph $H$, otherwise Painter
wins. The condition for Builder is that at the end the revealed graph
must belong to a certain class of graphs.  Deciding whether Builder
wins was recently shown to be PSPACE-complete~\cite{DvoVal15},
however, the proof assumes that some of the edges may be precolored.

Another recently studied asymmetric model is the
\textsc{Chooser-Picker Positional Game}. In it, the player Chooser
selects a pair of objects, and the player Picker selects one of them
for itself, leaving the other object for the player Chooser. The
winning condition is then similar to Maker-Breaker games, such as
\textsc{Online Ramsey Game} described above. A recent
paper~\cite{CMP16} proves NP-hardness for this problem, but
PSPACE-hardness remains open.

\boldsection{Proof outline.} The fact that \onlinechrom belongs to PSPACE is not hard to see:
The online coloring is represented by a game tree which is evaluated using the Minimax algorithm.
This can be done in polynomial space,
since the number of rounds in the game is bounded by $n$, i.e., the number of vertices,
and possible moves of each player can be enumerated in polynomial space: 
\P has at most $n$ possible moves, because it either uses a color already used for a vertex, or it chooses a new color,
and \D has at most $2^{s}$ moves where $s$ is the number of colored vertices,
since it chooses which colored vertices shall be adjacent to the incoming vertex.
\D must ensure that sent vertices form an induced subgraph of $G$,
but this can be checked in polynomial space.

Inspired by \cite{kudahl14}, we prove the PSPACE-hardness of \onlinechrom
by a reduction from {\sc Q3DNF-SAT}, i.e., the satisfiability of a
fully quantified formula in the 3-disjunctive normal form (3-DNF).  An
example of such a formula is

$$∀x_1 ∃x_2 ∀x_3 ∃x_4 …: (x_1 \wedge x_2 \wedge \neg x_3) \vee (\neg x_1\wedge x_2\wedge \neg x_4) \vee \dots$$

The similar problem of satisfiability of a fully quantified formula in
the 3-conjunctive normal form is well known to be PSPACE-complete. Since
PSPACE is closed under complement, {\sc Q3DNF-SAT} is PSPACE-complete
as well. Note that by an easy polynomial reduction, we can assume that
each 3-DNF clause contains exactly three literals.

We show the hardness in several iterative steps. First, in Section~\ref{sec:largePrecol}, we present a
new, simplified proof of the PSPACE-hardness of \onlineprecoloring in
which the sizes of both precolored and non-precolored parts of our
construction are linear in the size of the formula.

Then, in Section~\ref{sec:smallPrecol}, we strengthen the result by
reducing the size of the precolored part to be logarithmic in the size
of the formula.  This is achieved by adding linearly many
vertices to our construction.

Finally, in Section~\ref{sec:noPrecol}, we show how to remove one
precolored vertex and replace it by a non-precolored part, while
keeping the PSPACE-hardness proof valid. The cost for removing one
vertex is that the size of the graph is multiplied by a constant, but
since we apply it only logarithmically many times, we obtain a graph
of polynomial size and with no precolored vertex. This will complete the
proof of Theorem~\ref{thm:main}.

In our analysis, \P often uses the natural greedy algorithm \FF, which is ubiquitous
in the literature (see \cite{lst}, \cite{halldorsson00}):

\begin{definition}
The online algorithm \FF colors an incoming vertex $u$ using the smallest color not present %color with lowest index
among colored vertices adjacent to $u$.
\end{definition}

We remark that removing the last precolored vertex is the most difficult
part of proving PSPACE-hardness of \onlinechrom.
While there might be an easier way how to remove the penultimate
precolored vertex (using the last precolored vertex) and similarly
for previous precolored vertices,
for simplicity we use the same technique for removing
all precolored vertices as for removing the last precolored vertex.
Also, our technique can be used for any graph  %for removing a precolored vertex 
satisfying an assumption.

We note that if we would give \P some advantage like parallel edges,
precolored vertices (as in~\cite{kudahl14})
or something that helps \P distinguish different parts of the graph, the proof of PSPACE-hardness
would be much simpler. However, our theorem holds for simple graphs without any such help.

% We remark that we focus on the readability of the proof instead of optimizing the number of
% vertices in the resulting construction in the PSPACE-hardness proof provided that it
% has a polynomial size in comparison with the formula.

\section{Construction with a large precolored part}\label{sec:largePrecol}

Our first construction will reduce the PSPACE-complete problem
{\sc Q3DNF-SAT} to {\sc Online Coloring with Precoloring} with a large precolored part. Given a fully
quantified formula $Q$ in the 3-disjunctive normal form, we will
create a graph $G_1$ that will simulate this formula.
We assume that the formula contains $n$ variables $x_i, (1 ≤ i ≤ n)$
and $m$ clauses $C_a$, $(1 ≤ a ≤ m)$, and
that variables are indexed in the same order as they are quantified.

Our main resource will be a large precolored clique $K_{col}$ on $k$
vertices and naturally using $k$ colors; the number $k$ will be specified later.
Using such a precolored clique, we can restrict the allowed colors on a given
uncolored vertex $v$ by connecting it with the appropriate vertices in $K_{col}$,
i.e., we connect $v$ to all vertices in $K_{col}$ which do not have a color allowed for $v$.

For simplicity we use the precoloring in the strong sense, i.e., \P is
able to recognize which vertex in $K_{col}$ is which. We use this to
easily recognize colors.  However, it is straightforward to avoid the
strong precoloring by modifying the precolored part; for example by
creating $i$ independent and identical copies of the $i$-th vertex in
$K_{col}$, each having the same color and the same edges to other
vertices in $K_{col}$ and the rest of the graph. With such a
modification, \P would be able to recognize each color by the number of its
vertices in $K_{col}$. We also remark that working with
the strong precoloring is easier in the reduction,
and since we eventually obtain a graph without a precolored vertex,
it does not matter which precoloring we use.

Each vertex in $K_{col}$ thus corresponds to a color.
Colors used by \P are naturally denoted by numbers 1, 2, 3, ... $k$,
but we shall also assign meaningful names to them.

We now construct a graph $G_1$ that has the online chromatic
number $k$ if and only if the quantified 3-DNF formula can be
satisfied. See Figure~\ref{fig:gadgetexample} for an example of
$G_1$ and an overview of our construction. We use the following
gadgets for variables and clauses:

\begin{compactenum}\enumparams

\item For a variable $x_i$ which is quantified universally, we will
create a gadget consisting of \textit{\uvertices} $x_{i,t}$ and $x_{i,f}$,
connected by an edge. The vertex $x_{i,t}$ represents the positive literal $x_i$,
while $x_{i,f}$ represents the negative literal $\neg x_i$.
Both vertices have exactly two allowed colors:
$\t{set}_i$ and $\t{unset}_i$. If $x_{i,t}$ is assigned the color
$\set_i$, it corresponds to setting the variable $x_i$ to $1$, and
vice versa.

Note that if \D presents a vertex $x_{i,t}$ to \P,
\P is able to recognize that it is a vertex corresponding to the
variable $x_i$, but it is not able to recognize whether it is the vertex
$x_{i,t}$ or $x_{i,f}$.

\item For a variable $x_j$ which is quantified existentially, we will
create a gadget consisting of three \textit{\evertices}
$x_{j,t}$ (for the positive literal $x_j$), $x_{j,f}$ (for the literal $\neg x_j$)
and $x_{j,h}$ (the helper vertex), connected as a triangle.

Coloring of the first two variables also corresponds to setting the
variable $x_j$ to true or false, but in a different way: $x_{j,t}$ has
allowed colors $\set_{j,t}$ and $\unset_{j}$, while $x_{j,f}$ has
allowed colors $\set_{j,f}$ and $\unset_{j}$. We want to avoid both
$x_{j,t}$ and $x_{j,f}$ to have the color of type $\set$, and so the
``helper'' vertex $x_{j,h}$ can be colored only by $\set_{j,t}$ or
$\set_{j,f}$.

Note that the color choice for the vertices of $x_j$ means that if
\P is presented any vertex of this variable, \P can recognize it and decide
whether to set $x_j$ to $1$ (and color accordingly) or to $0$.

We call \evertices and \uvertices together \textit{variable vertices}.

\item For each clause $C_a$, we will add four
new vertices. 
\begin{compactenum}
\item First, we create a vertex $l_{a,i}$ for each literal in the clause,
which is connected to one of the vertices $x_{i,t}$ and $x_{i,f}$ corresponding to the sign of the literal.
For example if $C_a = (x_i ∧ ¬ {x_j} ∧ x_k)$, then $l_{a,i}$ is connected to $x_{i,t}$, $l_{a,j}$ is connected
to $x_{j,f}$ and $l_{a,k}$ to $x_{k,t}$. The allowed colors on a vertex $l_{a,i}$ are $\{f_a, \unset_i\}$.

\item Finally, we add a fourth vertex $d_a$ connected to the three vertices $l_{a,i}, l_{a,j}, l_{a,k}$.
This vertex can be colored only using the color $f_a$ or the color $\false_a$. The color
$\false_a$ is used to signal that this particular clause is evaluated to $0$.
If the color $f_a$ is used for the vertex $d_a$, this means that the clause is evaluated
to $1$, because $f_a$ is not present on any of $l_{a,i}, l_{a,j}, l_{a,k}$,
thus they have colors of type $\unset_i$ and their neighbors corresponding to literals
have colors of type $\set$.
\end{compactenum}

\item The last vertex we add to the construction will be $F$, a final vertex. The vertex
$F$ is connected to all the vertices $d_a$ corresponding to the clauses. The allowed colors
of the vertex $F$ are $\false_1, \false_2, \false_3, …, \false_{m}$. This final vertex
corresponds to the final evaluation of the formula. If all clauses are evaluated to $0$,
the vertex $F$ has no available color left and must use a new color.

\end{compactenum}

We have listed all the vertices and colors in our graph $G_1$ and the
functioning of our gadgets, but we will need slightly more edges. The
reasoning for the edges is as follows: If \D presents any vertex of
the type $l_{a,i}, d_a$ or $F$ before presenting the variable
vertices, or in the case when the variable vertices are presented out
of the quantifier order, we want to give an advantage to \P so it can
finalize the coloring.

This will be achieved by allowing \P to treat all remaining
\uvertices as \evertices, i.e., \P can
recognize which of the two \uvertices $x_{j,t}, x_{j,f}$ corresponds to
setting $x_j$ to $1$.

To be precise, we add the following edges to $G_1$:

\begin{compactitem}
\item Every \evertex $x_{j,t}, x_{j,f}, x_{j,h}$ is connected
to all previous \uvertices $x_{i,t}$, that is to all such
$x_{i,t}$ for which $i < j$.

\item Every \uvertex $x_{j,t}, x_{j,f}$ is connected
to all previous \uvertices $x_{i,t}$ such that $i < j$.
(We remark that these edges are not necessary, but make the following analysis simpler.)

\item Every vertex of type $l_{a,i}$ is connected to all the \uvertices
$x_{i't}$ for $i' ≠ i$. Note that $l_{a,i}$ is connected either to $x_{i,t}$, or to $x_{i,f}$;
we do not add a new edge to such vertices.

\item Every vertex of type $d_a$ is connected to all the \uvertices
$x_{i,t}$ for all $i$.

\item The vertex $F$ is connected to all the \uvertices
$x_{i,t}$ for all $i$.
\end{compactitem}

See Figure~\ref{fig:gadgetexample} for an example of our graph $G_1$.
We call all non-precolored vertices the \textit{gadgets} for variables and clauses.

\begin{figure}[!ht]
\centering\includegraphics[width=12cm]{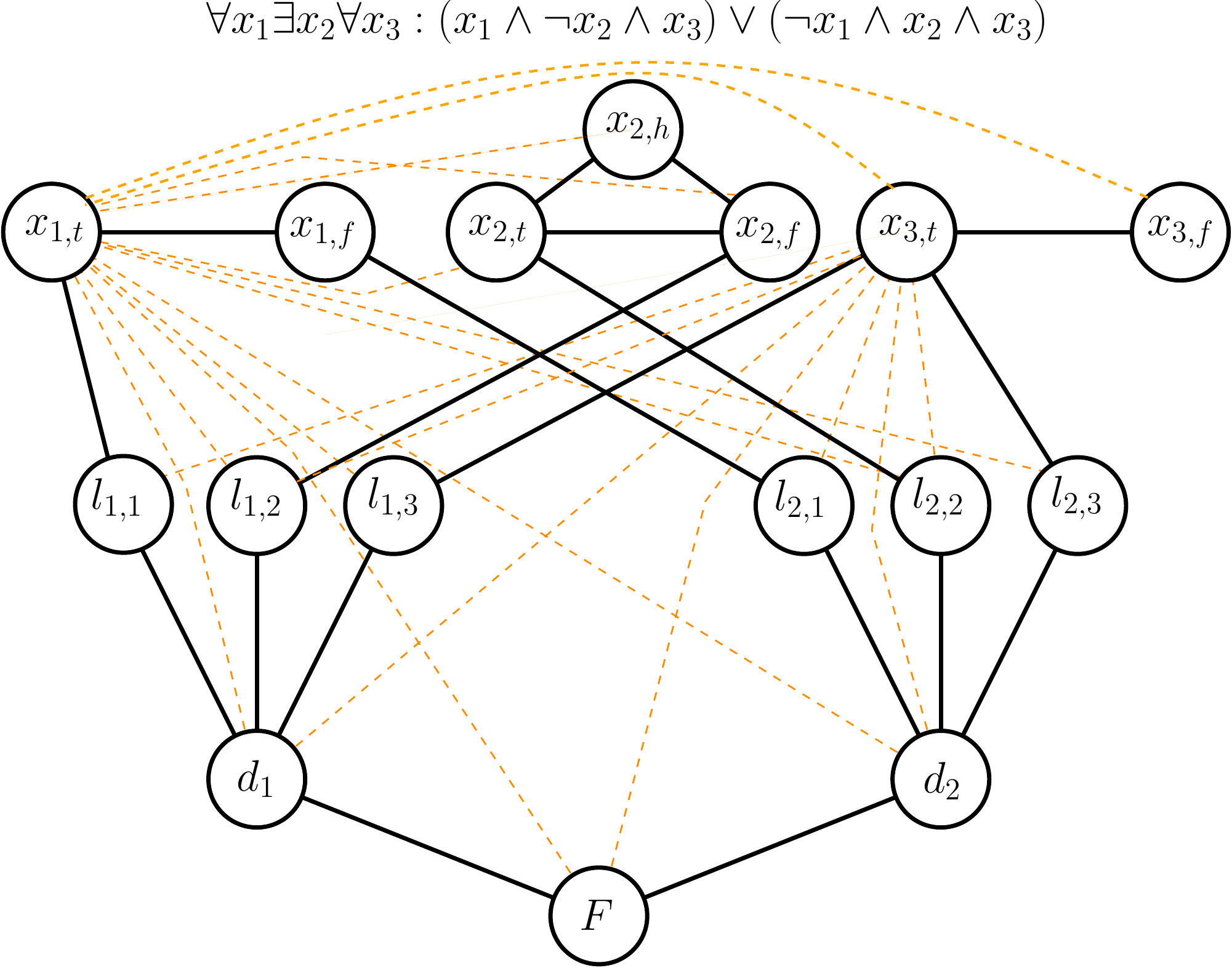}
\caption{The construction for a sample formula. The thick black edges are the normal
edges of the construction, and the dashed orange edges are the additional edges that
guarantee precedence of vertices. The lists of allowed colors of each vertex are not
listed in the figure.}
\label{fig:gadgetexample}
\end{figure}

It is easy to see that any two vertices outside $K_{col}$ have different sets of
allowed colors except \uvertices $x_{i,t}$ and $x_{i,f}$. 
%Also note that in any coloring, most colors are used only once
%outside $K_{col}$ with the exception of $\unset_i$ that may occur in vertices
%$l_{a,i}$ for each clause $a$ that contains the variable $x_i$ and of $f_a$ which
%may be used on at most three vertices $l_{a,i}$.
% PV: the second sentence is perhaps not needed

The number of colors allowed for \P
(the same as the size of $K_{col}$) is $k = 2m+2n_{\forall}+3n_{\exists}$
where $m$ is the number of clauses, $n_{\forall}$ the number of universally quantified variables
and $n_{\exists}$ the number of existentially quantified variables.

The next two lemmas contain the analysis of our construction.

\begin{lemma}\label{l:precolNotSat}
For a given fully quantified formula in the 3-DNF form
that is not satisfiable, \D can force \P to use $k+1$ colors.
\end{lemma}

\begin{proof}
If anytime during the game \P uses a color not
present in the color clique $K_{col}$, the lemma is proven. We
therefore assume this does not happen and show that the vertex $F$
cannot be colored with any color present in the precolored clique $K_{col}$.

\D's strategy is to first present the vertices $x_{*}$ in the
the order in which the variables are quantified in the formula. Whenever \D
sends a \uvertex $x_{i,t}$ or $x_{i,f}$, \P is not able to detect
whether it is setting the value of $x_i$ to $1$ or $0$. As the formula
is not satisfiable, \D can therefore present these vertices in
a sequence such that \P chooses the value of $x_i$ so that
the final evaluation is false.

After the vertices of the variables, \D will present the
vertices $l_{a,i}$ for each clause, then the vertex $d_a$ for each clause,
and finally the vertex $F$.

We now know that all clauses are evaluated to $0$. This means that at
least one of the vertices $l_{a,i}, l_{a,j}, l_{a,k}$ of each clause $a$ is assigned the
color $f_a$. It follows that the vertex $d_a$ of this clause must be
assigned the color $\false_a$. This holds for all clauses,
thus the vertex $F$ has a neighbor of the color $\false_a$ for each
clause $a$. Hence, $F$ cannot get any of the $k$ colors allowed and needs to be
assigned a new color, which completes the proof.
\end{proof}

\begin{lemma} \label{l:precolSat}
For a given fully quantified formula in the 3-DNF form
that is satisfiable, and for any order of sending vertices by \D, \P
has a strategy that uses $k$ colors. %does not use a new color for the vertex $F$.
\end{lemma}

\begin{proof}
\P's goal is to use the knowledge of a satisfiable evaluation
to paint the vertices with few colors. We note the following two observations.
The first follows easily from the sets of allowed colors.

\begin{observation}
If \D presents a vertex, \P is able to recognize it by the
set of edges to the clique $K_{col}$, i.e., by the set of allowed colors,
with the exception of \uvertices $x_{i,t}, x_{i,f}$ for a universally quantified variable $x_i$.
\end{observation}

\begin{observation}
If \D presents a vertex of type $l_{a,i}, d_a$ or $F$ before any
of the \uvertices $x_{i,t}, x_{i,f}$ is presented for a universally quantified variable $x_i$, \P will be able
to distinguish the vertex $x_{i,t}$ from the vertex $x_{i,f}$ and therefore
choose the assignment of $x_i$.

Furthermore, for \uvertices $x_{i,t}, x_{i,f}$
and for a vertex $v$ of a variable $x_j$ with $i<j$, %$x_{j,t}, x_{j,f}$ or $x_{j,h}$ for $i<j$,
if \D presents $v$ before any of these \uvertices is sent, \P
is then able to distinguish the \uvertices $x_{i,t}$ and $x_{i,f}$.
\end{observation}

The second observation is easy to see by noting that presenting $v, l_{a,i}, d_a$ or $F$
means that it has an edge with $x_{i,t}$, but not with $x_{i,f}$
(or, in the case of $l_{a,i}$, it may have an edge with $x_{i,f}$, but not
with $x_{i,t}$, but this is symmetric). This means that we can
distinguish $x_{i,t}$ from $x_{i,f}$ by the presence or absence of an
edge.

We continue with the proof of the lemma.
We say that a variable $x_i$ is \textit{set}, if at least one of its vertices was colored
or if \P has assigned a value to it
(e.g., if a vertex of a variable $x_j$ was sent before a vertex of $x_i$ for $j > i$);
otherwise $x_i$ is \textit{unset}.
Until \D presents a vertex of type $l_{a,i},
d_a$, or $F$, \P colors an incoming vertex $v$ that corresponds to a variable $x_i$
using the following strategy:
\begin{compactitem}
\item Let $U$ be the set of all unset variables $x_j$ with $j<i$, i.e.,
all previous unset variables.
\P chooses an assignment for all variables in $U$
according to satisfiability of the formula.
Then \P can color vertices for variables in $U$ according to this assignment, 
since it is now able to distinguish \uvertices of each universally quantified variable in $U$.
Therefore every variable in $U$ becomes set.

\item If the variable $x_i$ is not set and it is quantified universally,
\P chooses an arbitrary allowed color for $v$ ($\set_i$ or $\unset_i$),
thus $x_i$ becomes set.

\item If $x_i$ is not set and it is quantified existentially,
\P knows how to set this variable to satisfy the formula,
thus \P colors $v$ according to the value of $x_i$, since $v$ can be recognized.

\item If $x_i$ is set and \P can recognize which vertex is $v$,
it colors $v$ according to the setting of $x_i$.

\item If $x_i$ is set and \P cannot recognize $v$,
then $x_i$ must be quantified universally and the other vertex for $x_i$ is colored,
therefore there is a single color left for $v$.
(Note that in this case no vertex for a variable $x_j$ with $j>i$ arrived.)
\end{compactitem}

When the vertex $u$ of type $l_{a,i}, d_a$, or $F$ is sent,
let $U$ be the set of all unset variables. 
\P chooses an assignment for all variables in $U$
according to satisfiability of the formula.
Then \P can color vertices for variables in $U$ according to this assignment, 
since it is now able to distinguish vertices of each universally quantified variable in $U$.
Then all variables are set and \P decides how to color all remaining vertices in the graph
using the following rules:
\begin{compactitem}
\item A variable vertex obtains color according to the setting of the corresponding variable.

\item A vertex of type $l_{a,i}$ gets $\unset_i$ if its adjacent variable vertex $x_{i,t}$ or $x_{i,f}$
has a color of type $\set$; otherwise it gets $f_a$

\item A vertex of type $d_a$ obtains color $f_a$
if all of its adjacent vertices of type $l_{a,i}$ have colors of type $\unset$,
i.e., the clause $a$ evaluates to $1$;
otherwise it obtains color $\false_a$.

\item The final vertex $F$ obtains a color $\false_a$ for a clause $a$ that
evaluates to $1$. Since \P set variables such that the formula is satisfied, there
must be such clause. 
\end{compactitem}

Now, \P colors an incoming vertex $v$ with the color
assigned to it using these rules. Hence, no matter in which order
\D sends the vertices, the graph is colored using $k$ colors as desired.
\end{proof}

\section{Construction with a precolored part of logarithmic size}\label{sec:smallPrecol}

We now make a step to the general case without precoloring by reducing the size
of the precolored part so that it has only logarithmic size.
Our construction is based on the one with
a large precolored part; namely, all the vertices $x_{i,t}, x_{i,f}, x_{j,t}, x_{j,f},
x_{j,h}, l_{a,i}, d_a, F$ (the gadgets for variables and clauses)
and the whole color clique $K_{col}$ will be connected the same way.
Let $G_1$ denote the gadgets for variables and clauses and $K_{col}$.

Since $K_{col}$ now starts uncolored and \D may send it after the gadgets,
we help \P by a structure for recognizing vertices in $G_1$ % the gadgets and in $K_{col}$
or for saving colors.

We remark that there is also a simpler construction with a logarithmic number of precolored vertices.
If we just add a clique of logarithmically many precolored vertices to recognize vertices in $G_1$, 
the following proof would work and be easier. However, when we replace a precolored vertex $v$ by some non-precolored graph
in Section~\ref{sec:noPrecol}, we will use some conditions that this simple construction would
not satisfy. Namely, we shall need that precolored vertices are not adjacent to $G_1$.
%  that the vertices connected to $v$ may be colored by \P using \FF before %a greedy strategy
% two nonadjacent vertices from the rest of the graph arrive.

\subsection{Nodes}

Our structure will consist of many small \textit{nodes}, all of them have
the same internal structure, only their adjacencies with other vertices vary.

Each node consists of three vertices and a single edge;
vertices are denoted by $p_1, p_2, p_3$ and the edge leads between $p_2$ and $p_3$.
We call the vertices $p_1$ and $p_2$ the \textit{lower partite set} of the node,
$p_3$ form the \textit{upper partite set}. See Figure~\ref{fig:node} for an illustration of a node.
\begin{wrapfigure}{r}{2.8cm} %[!ht]
\includegraphics[width=2cm,right]{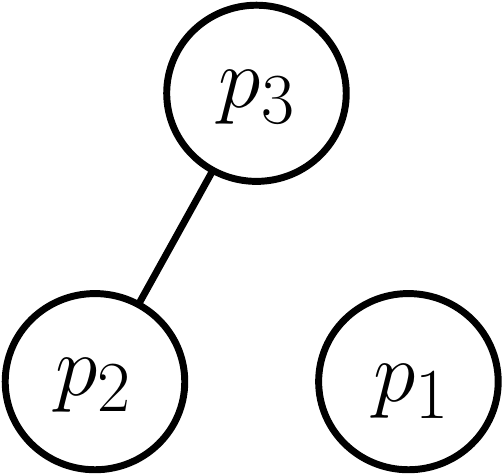}
\caption{Node}
\label{fig:node}
\end{wrapfigure}
Clearly, the online chromatic number of a node is two.

The intuition behind the nodes is as follows:

\begin{compactitem}

\item If \D presents vertices of a node in the correct
way, \P needs to use two colors in the lower partite set of the node.

\item No color can be used in two different nodes.

\item Each vertex $v\in G_1$ (in the gadgets and in $K_{col}$) has its own associated node $A$. 
If the vertex $p_3$ from $A$ does not arrive before $v$ is sent,
\P can color $p_3$ and $v$ with the same color, thus saving a color.
Otherwise, \P can use the node to recognize $v$.

\end{compactitem}

\begin{compactitem}
\item \uvertices $x_{i,t}, x_{i,f}$ for each universally quantified variable $x_i$
should be distinguishable only by the same vertices as in the previous section.
Therefore they are both associated with the same two nodes.

\end{compactitem}

Let $N$ be the number of vertices in $G_1$.
We create $N$ nodes, denoted by $A_1, \dots, A_{N}$, one for each vertex in $G_1$.
For any two distinct nodes $A_i$ and $A_j$ ($i\neq   j$), there is an edge between each vertex in $A_i$ and
each vertex in $A_j$. Therefore, no color can be used in two nodes.

We have noted above that each node is associated with a vertex; we now make the 
connection precise. Let $v_1, \dots, v_N$ be the vertices in $G_1$ (in an arbitrary order).
Then we say that $A_i$ \emph{identifies} the vertex $v_i$. Moreover, if $v_i$ is a vertex 
$x_{\ell,t}$ or $x_{\ell,f}$ for a universally quantified variable $x_\ell$
and $v_j$ is the other vertex, then $A_j$ also identifies $v_i$ and $A_i$ also identifies $v_j$.
Thus each node identifies one or two vertices and each vertex is identified by one or two nodes.

Edges between a vertex $v$ in the original construction $G_1$ and a node depend on
whether the node identifies $v$, or not.
For a vertex $v\in G_1$ and for a node $A$,
if $A$ identifies $v$, we connect only the whole lower partite set of $A$ to $v$, i.e.,
we add two edges from $v$ to both $p_1$ and $p_2$ of $A$.
Otherwise, we add three edges -- one between $v$ and every vertex in $A$.
%TODO Add a picture?

\subsection{Precolored vertices}
The only precolored part $P$ of the graph is intended for distinguishing nodes.
Since there are $N$ nodes in total, we have $p = \lceil \log_2 N\rceil$ precolored vertices
$z_1, z_2, \dots z_p$ with no edges among them. %Note that $p = \O(log N)$.
Precolored vertices have the same color that may be used later for coloring $G_1$ (the gadgets and $K_{col}$).
For simplicity, we again use the precoloring in the strong sense, i.e., \P is able
to recognize which precolored vertex is which.

We connect all vertices in the node $A_i$ to $z_j$ if the $j$-th bit in the binary
notation of $i$ is 1; otherwise $z_j$ is not adjacent to any vertex in $A_i$.

Clearly, the node to which an incoming vertex belongs can be recognized by its adjacency to the precolored vertices.
Note that a vertex from nodes is connected to at least one precolored vertex
and there is no edge between $G_1$ and precolored vertices.

\medskip

So far, we have introduced all vertices and edges in our construction
of the graph $G_2$.
To summarize, our graph $G_2$ consists of three \textit{parts}:
\begin{compactenum}
\item\label{vertype:1} The graph $G_1$ from the previous section
consisting of the color clique $K_{col}$ and
of the gadgets for variables and clauses, i.e.,
vertices $x_{i,t}, x_{i,t}, x_{j,t}, x_{j,f}, x_{j,h}, l_{a,i}, d_a, F$;
\item\label{vertype:2} the nodes for recognizing vertices in $G_1$; %the gadgets and in $K_{col}$;
\item\label{vertype:3} the precolored vertices $P$ for recognizing the nodes.
\end{compactenum}

See Figure~\ref{fig:logarithmic} for a simplified example of the graph $G_2$.
\smallskip

\begin{figure}[!ht]
\centering\includegraphics[width=10cm]{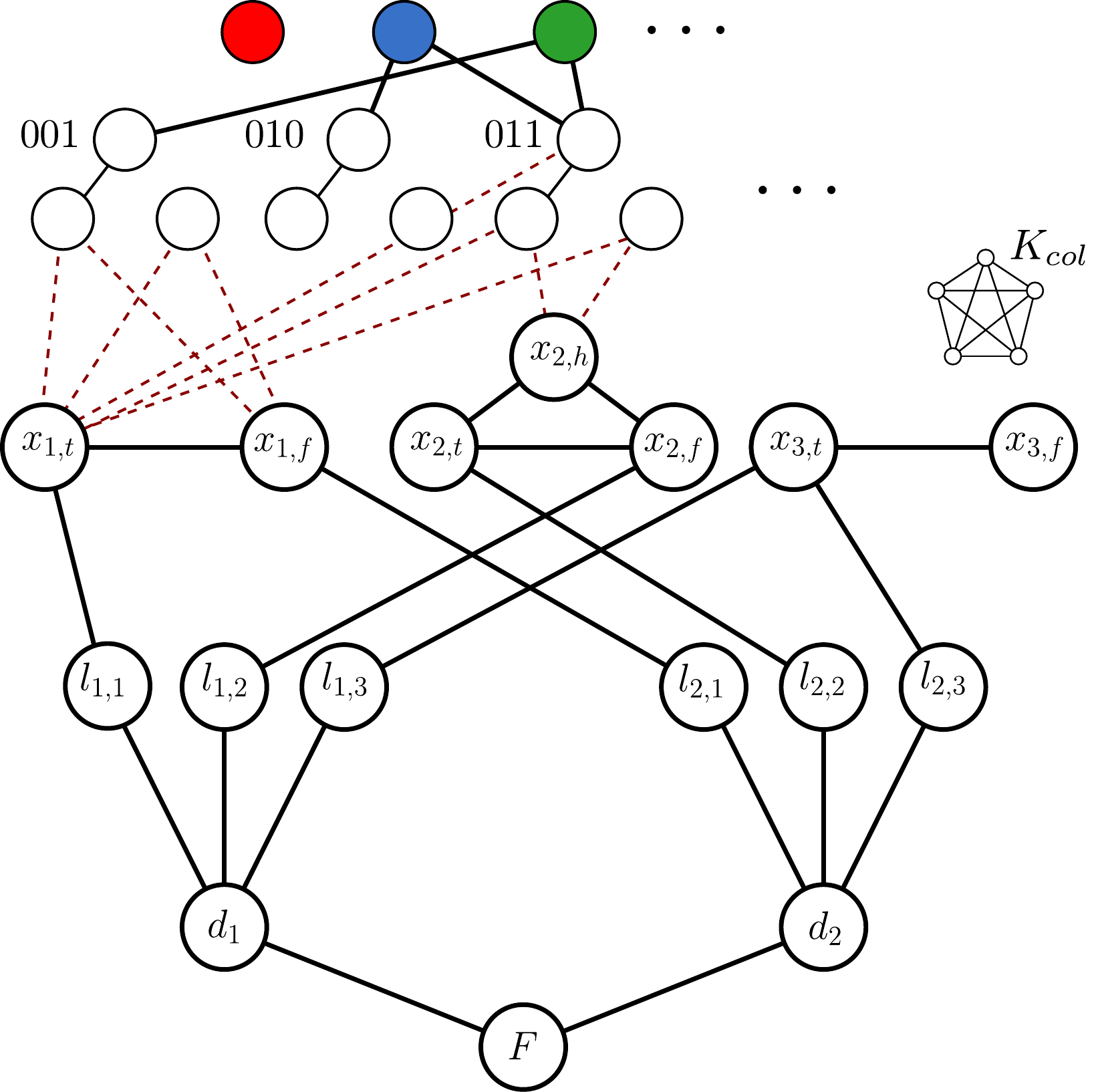}
\caption{A visualisation of the construction with logarithmically many precolored vertices.
Not all edges of the construction are present, e.g., dashed edges from Figure~\ref{fig:gadgetexample},
most edges between nodes and $G_1$, and precolored vertices are connected only to $p_3$ of a displayed node,
although they are connected either to all vertices in a node, or to none of them. The now uncolored clique $K_{col}$
is connected the same way as it was in Section~\ref{sec:largePrecol}.
Notice that $x_{1,t}$ and $x_{1,f}$ are identified by the first node (as they are connected
only to $p_1$ and $p_2$ of the first node), but $x_{1,t}$ is not identified by the third node
(as it is connected to all vertices of the third node).}
\label{fig:logarithmic}
\end{figure}

We start with an observation about nodes.

\begin{observation}\label{o:eachNode2Colors}
If \D reveals all the nodes $A_i$
before it sends any vertex from $G_1$ (the gadgets and $K_{col}$),
\D can force \P to use two colors in the lower partite set of each node
and these colors are different for each node.
\end{observation}

\begin{proof}
First we see that \P is not able to distinguish incoming vertices from a node $A$
by their edges to vertices in other nodes, since there is a complete bipartite
graph between any two nodes.

Thus \D uses the following strategy for each node $A$ independently:
\D first presents the vertex $p_1$ from $A$;
let $c$ be its color. Then \D sends a vertex $q$ that is one of vertices $p_2$ and $p_3$.
\P cannot deduce which of them is $q$, because
both are not connected to $p_1$ (and both are connected to all vertices in other nodes). % and no other vertex was sent.
If \P assigns the color $c$ to $q$, then $q=p_3$ and
\D sends $p_2$ which must get another color than $c$.
Otherwise if $q$ obtains another color than $c$, then $q=p_2$.
%, \D sends $p_3$ and there are also two colors in the lower partite set.
\end{proof}

\P has $2N$ colors for the $N$ nodes and $k$ colors for $G_1$ (with the same names as in the previous section),
thus overall \P is allowed to use $k' = 2N+k = 2N + 2m+2n_{\forall}+3n_{\exists}$ colors
where $m$ is the number of clauses, $n_{\forall}$ the number of universally quantified variables
and $n_{\exists}$ the number of existentially quantified variables.
The precolored vertices have a color that may be used later in $G_1$. % the gadgets or in $K_{col}$.

\begin{lemma} \label{l:smallPrecolNotSat}
For a given fully quantified formula $Q$ in the 3-DNF form
that is not satisfiable, \D can force \P to use $k'+1$ colors.
\end{lemma}

\begin{proof}
If at any point of the game there are $k'+1$ colors in the graph,
the lemma is proven. %We therefore assume this does not happen and show that the vertex $F$
%cannot be colored with any of $k'$ allowed colors.
\D's strategy is to first present all
nodes (in any order) and force to use two colors in the lower partite set
of each node. %, i.e., on vertices $p_1$ and $p_2$.
Moreover, \P has to use different colors in distinct nodes.
Forcing such a coloring is possible by Observation~\ref{o:eachNode2Colors}.

Then \D sends $K_{col}$ and the situation is similar to the one with the precolored $K_{col}$,
since none of two colors used in the lower partite set of a node is allowed for a vertex in $G_1$,
thus \P must use the $k$ colors in $K_{col}$ for coloring the gadgets.
\P is able to recognize vertices in $K_{col}$ by nodes, but nodes do not give \P
additional knowledge compared to $K_{col}$. In particular, \P is not able to distinguish
\uvertices $x_{i,t}, x_{i,f}$ for a universally quantified variable $x_i$ when one of them arrives.
We conclude the proof by applying Lemma~\ref{l:precolNotSat}.
% MB: We do not need to say we use just the proof, pretty much the lemma can be used
\end{proof}

For a satisfiable formula nodes and precolored vertices become important.
We give another useful observation about nodes.

\begin{observation} \label{o:nodesTwoColoredByGreedy}
Suppose that \P is using \FF,
%before a vertex from the gadgets or from $K_{col}$ arrives,
i.e., \P always assigns the smallest color not present
among colored vertices adjacent to the incoming vertex.
Then there are at most two colors used on each node.

Moreover, if vertices $p_1$ and $p_2$ from a node $A$ arrive before $p_3$ from $A$,
then $p_1$ and $p_2$ have the same color.
\end{observation}

\begin{proof}
Consider the last vertex $q$ from a node $A$ that is sent from $A$. We distinguish three cases:
\begin{compactitem}
\item $q = p_1$: Let the color of $p_2$ be $c$ and the color of $p_3$ be $d$ (clearly $c\neq d$).
Since the edges from $p_1$ and $p_2$ to other vertices in the construction except $p_3$
are exactly the same, \P can use $c$ for $p_1$, but
it cannot use any color forbidden for $p_2$ when $p_2$ was colored with the possible exception of $d$.
Hence $p_1$ obtains $c$ or $d$, but in both cases the node $A$ has two colors.
\item $q = p_2$: Let the color of $p_1$ be $c$ and the color of $p_3$ be $d$.
If $c\neq d$, \P uses $c$ for $p_2$ for the same reason as in the previous case.
Otherwise $c=d$ and \P must color $p_2$ with another color, but there are two colors on $A$.
\item $q = p_3$: We show that $p_1$ and $p_2$ have the same color. 
Without loss of generality, $p_1$ arrives first and obtains a color $c$.
When $p_2$ arrives, it can be colored by $c$, but by no other color $c' < c$,
since $c'$ would be available for $p_1$ also when $p_1$ was colored, because
the edges from $p_1$ and $p_2$ to other vertices except $p_3$ are the same.
Hence $p_1$ and $p_2$ have the same color and the lemma follows.
\end{compactitem}
\end{proof}

\P's strategy to win is basically the following: Color vertices in nodes greedily
until a vertex $u$ from $G_1$ arrives.
At this time, some (but maybe not all) vertices in the gadgets can be recognized by their nodes. \P uses the winning strategy from Lemma~\ref{l:precolSat}
on the gadgets for vertices that it can recognize, even if $K_{col}$ has arrived only partially. % or at all.
For each vertex $v$ that cannot be recognized, \P is able to color (or already colored) the lower partite set of the $v$'s node $A$
with only one color, therefore it can use the same color for $v$ and the vertex $p_3\in A$,
since they are not adjacent.

In the following proof, \P waits for two nonadjacent vertices in $G_1$,
even although one vertex from $G_1$ suffices. The reason is the we shall
need such a condition in Section~\ref{sec:noPrecol} when we remove precolored vertices.

\begin{lemma} \label{l:smallPrecolSat}
For a given fully quantified formula $Q$ in the 3-DNF form
that is satisfiable, and for any order of sending vertices by \D, \P
has a strategy that uses $k'$ colors.
\end{lemma}

\begin{proof}
Let $\C$ be the set of $2N$ colors for nodes
and let $\calD$ be the set of $k$ colors for $G_1$ (the gadgets and $K_{col}$)
including the color used on precolored vertices; both $\C$ and $\calD$ are ordered arbitrarily
and $\C \cap \calD = \emptyset$.

\P utilizes the precolored part $P$ to decide whether an incoming vertex belongs 
to $G_1$ (recall that $G_1$ is the part of $G_2$ not adjacent to any precolored vertex),
or to a node (and to which node). First, \P uses the following algorithm:

\algobox{
\textbf{Algorithm~\GR:}
For an incoming vertex $u$ sent by \D:
\begin{compactenum}
\item If there are two nonadjacent vertices $u_1$ and $u_2$ in $G_1$ that arrived,
stop the algorithm.
\item If $u$ is from nodes, assign $u$ the smallest color from $\C$ %with the lowest index
not present among colored neighbors of $u$.
\item Otherwise, if $u$ is from $G_1$, assign $u$ the smallest color from $\calD$
not present among colored neighbors of $u$.
\end{compactenum}
}

Note that the last $u$ considered by \GR is not yet colored and $u$ is from $G_1$.
Observe that only a clique is colored in $G_1$
and this clique has at most $k$ vertices.

%TODO maybe not needed
Let $u_1$ and $u_2$ be the two nonadjacent vertices from the stopping condition of \GR.
Observe that the nodes identifying $u_1$ and $u_2$ are different,
since these nodes can be the same only for both \uvertices of one variable,
but these vertices are adjacent.

\P continues by the following algorithm. We remark that by ``\P can recognize $u\in G_1$''
we mean that \P knows which vertex in $G_1$ corresponds to $u$ unless
$u$ is one of the two vertices for a universally quantified variable $x_i$ --
then \P knows that one of $u = x_{i,t}$ and $u = x_{i,f}$ holds.

\algobox{
\textbf{Algorithm~\WIN:}
\P creates a virtual graph $\virt{G}$ with all colored vertices in $G_1$;
the colors of such vertices are inherited from the graph $G_2$.

\hskip 10pt \P shall simulate its winning strategy on $\virt{G}$ using colors from $\calD$.
Since only a clique is colored in $\virt{G}$, it renames colors
so that the colored vertices have colors according to the winning strategy on $\virt{G}$.
\P remembers the colors of vertices in the virtual graph $\virt{G}$. % during the whole run of the algorithm.

\medskip

For an incoming vertex $u$ sent by \D:
\begin{compactenum}
\item If $u$ is from $G_1$:

\item \label{alg-ln:WIN-precol-G1vertRecognized} \indentskip If there is a vertex $v$ in nodes that is not adjacent to $u$, % 
then $v$ is from one of at most two nodes identifying $u$ and \P can recognize $u$.
A virtual \D sends $u$ to $\virt{G}$ and let $c\in \calD$ be the color that it obtains
by the winning strategy of \P on $\virt{G}$ by Lemma~\ref{l:precolSat}.
Color $u$ in $G_2$ using $c$ and call $u$ \textit{recognized}.
\item \label{alg-ln:WIN-precol-G1vertNotRecognized} \indentskip Otherwise, assign $u$ the smallest color from $\C$
not present among colored neighbors of $u$ and not used in $G_1$.

\item If $u$ is from a node $A$:

\item \label{alg-ln:WIN-precol-nodesSaveColor} \indentskip If a vertex $v$ from $G_1$ identified by $A$ is colored by $c$,
there is no edge between $v$ and $u$, $c\in \C$ and $c$ is not present among nodes,
then color $u$ using $c$. 
\item \indentskip Otherwise, assign $u$ the smallest color from $\C$
not present among colored neighbors of $u$.
\end{compactenum}
}

If $u$ is from $G_1$ and \P can recognize it by a nonadjacent vertex $v$ from nodes (line~\ref{alg-ln:WIN-precol-G1vertRecognized}),
then coloring it with the winning strategy is correct, since vertices that have a color from $\calD$ in $G_2$
have the same color in $\virt{G}$.

Otherwise, if \P cannot recognize $u$ from $G_1$ (line~\ref{alg-ln:WIN-precol-G1vertNotRecognized}),
the vertex $p_3$ from each node $A$ that serves for identifying $u$ is not yet sent.
It follows that there is at most one color used in the lower partite set of $A$
by Observation~\ref{o:nodesTwoColoredByGreedy}.
In this case, $u$ gets some color $c$ from $\C$ that is not used in $G_1$.
We want to show that $c$ will be used on the vertex $p_3$ from a node that identifies $u$.

To see this, observe that when the vertex $p_3$ from a node $A$ that identifies $u$ arrives,
the condition on line~\ref{alg-ln:WIN-precol-nodesSaveColor} is satisfied and $p_3$ obtains $c$ unless
the color $c$ is already used in nodes. 
Thus if the unrecognized vertex $u$ from $G_1$ 
is not a \uvertex, it has the same color as the vertex $p_3$ from the single node $A$ identifying $v$,
since $c$ can be used in nodes only for $p_3$ of the node $A$.

Otherwise, if $u$ is an unrecognized \uvertex, then $u$ is identified by two nodes $A$ and $A'$.
Both vertices $p_3$ from $A$ or $A'$ arrive after $u$ and one of them obtains $c$,
since the algorithm prefers to use colors of $v$ and possibly of the other \uvertex
for the same variable (if its color is in $\C$).

%First observe that if $u$ is from a node $A$, a vertex $v$ from $G_1$ identified by $A$ is colored by $c$,
%there is no edge between $v$ and $u$, $c\in \C$ and $c$ is not present among nodes (line~\ref{alg-ln:WIN-precol-nodesSaveColor}),
%then $u$ is the vertex $p_3$ of the node $A$.
%Also, there is at most one color used in the lower partite set of $A$ before coloring $u$
%by Observation~\ref{o:nodesTwoColoredByGreedy}
%and $c$ is not used anywhere except $v$, because the algorithm colored $v$ such that $c$ was not used in $G_1$. 
%Thus the algorithm can use $c$ for $u$.
%
%An unrecognized vertex $v$ from $G_1$ which 
%is not a \uvertex
%has the same color as the vertex $p_3$ from the single node $A$ identifying $v$, 
%since $A$ identifies only $v$, other vertices in $G_1$ are adjacent to $p_3$, and
%$p_3$ arrives after $v$, thus the condition on line~\ref{alg-ln:WIN-precol-nodesSaveColor} is satisfied when
%the algorithm colors $u$.
%
%An unrecognized \uvertex $v$ identified by two nodes $A$ and $A'$
%obtains some color $c\in \C$.
%%Let $d$ be the color of the other vertex $w$ for the same universal variable.
%Both vertices $p_3$ from $A$ or $A'$ arrive after $v$ and one of them obtains $c$,
%since the algorithm prefers to use colors of $v$ and the other \uvertex
%for the same varible (if its color is in $\C$), similarly as in the previous paragraph.

Recognized vertices from $G_1$ are colored by the wining strategy
on $\virt{G}$ using $k$ colors from $\calD$. As we have shown above, each unrecognized vertex $v$ has the 
same color as a vertex $p_3$ in one of the nodes identifying $v$. There are at most $2N$
colors used in nodes by Observation~\ref{o:nodesTwoColoredByGreedy} and all of them are from $\C$.
%Since $|\calD| = k$, $|\C| = 2N$ and $\calD \cap \C = \emptyset$,
Thus \P uses at most $k' = 2N+k$ colors.
\end{proof}

\section{Removing precoloring}\label{sec:noPrecol}

In this section we show how to replace one precolored vertex by a large nonprecolored graph
whose size is a constant factor of the size of the original graph, while keeping \P's winning strategy 
in the case of a satisfiable formula. \D's winning strategy in the other case is of course
preserved as well and easier to see.
We prove the following lemma which holds for all graphs with precolored vertices satisfying
an assumption.

\begin{lemma}\label{l:removePrecolVertex}
Let $G$ be a graph with precolored subgraph $G_p$ %created from a fully quantified formula $\phi$,
%let  be the subgraph of its precolored vertices, 
and let $v_p\in G_p$ be a precolored vertex of $G$. % and 

Let $D$ be the induced subgraph with all nonprecolored vertices that are \textit{not} connected to $v_p$
and let $E$ be the induced subgraph with all nonprecolored vertices that are connected to $v_p$.

Let $k$ be an integer such that if $\chi^O(G)\leq k$, then
there exists a winning strategy of \P where %from Lemma~\ref{l:smallPrecolSat} -- no reference to lemma, especially if it is in appendix
\P colors $E$ using \FF before two nonadjacent vertices from $D$ arrive. Moreover,
in this case if \FF assigns the same color to a vertex in $D$ and to a vertex in $E$ before two nonadjacent vertices from $D$ arrive,
\P can still color $G$ using $k$ colors.
%TODO check this condition, make it precise

Then there exists an integer $k'$ and a graph $G'$ with the following properties:
%\smallskip

\begin{compactitem}
\item $G'$ has only $|V(G_p)| - 1$ precolored vertices, % (where $|V(G_p)|$ is the number of vertices in $G_p$),
and $|V(G')| \leq 25 |V(G)|$,
\item $G'$ can be constructed from $G$ in polynomial time with respect to the size of $G$,
\item it holds that $\chi^O(G')\leq k'$ if and only if $\chi^O(G)\leq k$.
\end{compactitem}
\end{lemma}

Note that we do not assume anything about $D$ or $E$ except that $E$ may be colored by \P using \FF
before two nonadjacent vertices from $D$ arrive. %\TODO{More remarks about the statement of the lemma?}
We give a proof of Theorem~\ref{thm:main} using Lemma~\ref{l:removePrecolVertex} in
Section~\ref{sec:proofOfThm}.

\boldsection{Construction of $G'$.} Let $N$ be the total number of vertices in $D$ and $E$
and let $S = 8N$. %TODO S=4N might be enough
Our graph $G'$ consists of precolored part $G_p' := G_p \setminus \{v_p\}$, graphs $D$ and $E$ and three huge cliques $A, B$ and $C$ of size $S$; cliques $A, B$ and $C$ together form a \textit{supernode}.
We keep the edges inside and between $D$ and $E$ and the edges between $G_p'$ and $D \cup E$ 
as they are in $G$.

We add a complete bipartite graph between cliques $B$ and $C$,
i.e., $B\cup C$ forms a clique of size $2S$.
No vertex in $A$ is connected to $B$ or $C$.
In other words, the supernode is created from a node by replacing each vertex by a clique of size $S$
and the only edge in the node by a complete bipartite graph.

There are no edges between the supernode (cliques $A$ and $B\cup C$) and any precolored vertex in $G_p'$. 
It remains to add edges between the supernode and $D\cup E$.
There is an edge between each vertex in $E$ and each vertex in the supernode,
while every vertex in $D$ is connected only to the whole $A$ and $B$,
but not to any vertex in $C$. The fact that $D$ and $C$ are not adjacent at all is essential in our analysis.
Our construction is depicted in Figure~\ref{fig:supernode}.

\begin{figure}[!ht]
\centering\includegraphics[width=4cm]{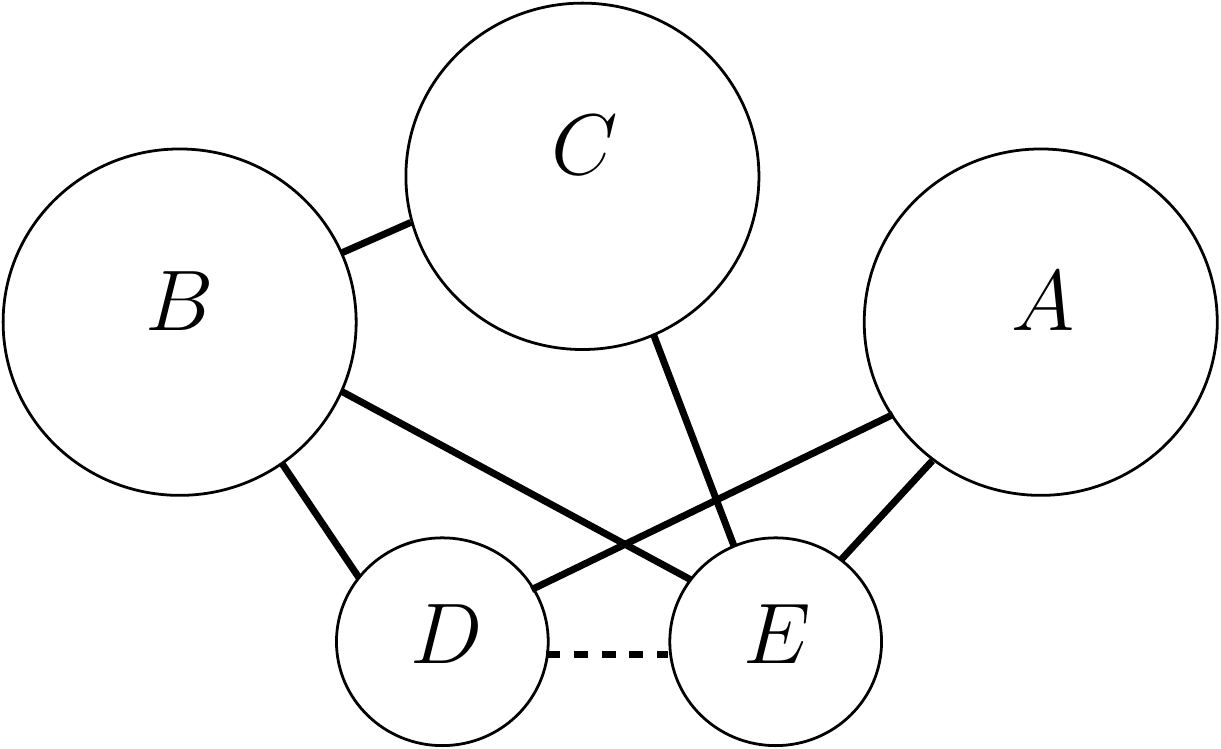}
\caption{An illustration of our construction of $G'$ where parts $A, B$ and $C$ form a supernode
and parts $D$ and $E$ form the original graph $G$ (the remaining precolored vertices are not shown).
A thick line connecting two parts of $G'$ corresponds to a complete bipartite graph between them.
The edges between $D$ and $E$ remain unchanged from the original graph $G$.
If there is no line between two parts, then there is no edge between them in $G'$.}
\label{fig:supernode}
\end{figure}

\begin{proof}[Proof of Lemma~\ref{l:removePrecolVertex}] Let
$G'$ be the graph defined as above.
Note that the number of vertices in $G'$ is at most $25 |V(G)|$,
$G'$ can be constructed from $G$ in polynomial time
and $G'$ has only $|V(G_p)| - 1$ precolored vertices. Therefore,
it remains to prove $\chi^O(G')\leq k'$ for some $k'$ if and only if $\chi^O(G)\leq k$.
We set $k'$ to $k + 2S$, since at most $2S$ colors will be used in the supernode.

We start with the ``only if'' direction which is easier.
Suppose that $\chi^O(G) > k$ and consider $G'$ constructed from $G$.
\D starts by sending $A$ and then $B\cup C$ such that no color is both in $A$ and in $B$;
this is possible, since cliques $A$ and $C$ have the same size.
More precisely, if \P would assign a color from $A$ to a vertex from $B\cup C$,
\D decides that the vertex is in $C$; otherwise the vertex is in $B$ (or in $C$ if every vertex in $B$ has been colored).
This forces \P to use $2S$ different colors on $A$ and $B$ and no such color can be used for $D$ and $E$.

Now the colored part of $G'$ (precolored vertices and the supernode)
may be viewed as a precolored part in $G$ except that \P
may not be able to use the supernode in the same way as a precolored vertex. Since this only helps \D,
\D sends vertices according
to its winning strategy for $G$. This proves the ``only if'' direction
of the third proposition of Lemma~\ref{l:removePrecolVertex}.

In the rest of this section we focus on the opposite direction:
assuming that $\chi^O(G)\leq k$, we show that \P can color $G'$
with $k'$ colors regardless of the strategy of \D.

In the following, when we refer to the colored part of $G'$, we do not
take precolored vertices into account.  \P actually does not look at
precolored vertices unless it uses its winning strategy for coloring
$G$ with $k$ colors.

\boldsection{Intuition.}
At the beginning \P has too little data to infer anything about the
vertices.  Therefore, \P shall wait for two nonadjacent vertices from
$D$ and for two large cliques (larger than $S/2$) with a small
intersection. Before such vertices arrive, it will color greedily.

Note that the greedy coloring algorithm eventually stops before everything is colored.
Having two large cliques, one mostly from $A$ and the other mostly from $B\cup C$,
and two nonadjacent vertices from $D$,
\P is able to recognize where an incoming vertex belongs.
Therefore, \P can use the supernode like a precolored vertex and colors the remaining vertices from $D$ and $E$ by
its original winning strategy on $G$.

This approach may fail if a part of $D$ is already colored by \P's
application of the greedy strategy. To remedy this, we prove that
colors used on $D$ so far are also used in $C$ or $E$, or will be used
on $C$ later.

The other obstacle is that \P might not be able to
distinguish between one clique from $D$ and vertices in $A$
if nothing from $B$ arrives. Nevertheless, each vertex $u$ in such
a ``hidden'' clique is connected to all other colored vertices
in $D$ and to the whole colored part of $E$, otherwise it would be distinguishable
from vertices in $A$. Hence, it does not matter which color $u$ receives
(since it is universal to colored vertices in $D$ and $E$).

In summary, the sheer size of the supernode should allow the
player \P to be able to use it as if it were precolored. Still,
we need to allow for some small margin of error. This leads us to the
following definition:

\begin{definition}\label{dfn:almost}
Let $N$ be the number of vertices of $D \cup E$ as in the construction of $G'$. For subgraphs $X, Y \subseteq G'$,
we say that $X$ is \textit{practically a subgraph} of $Y$ if $|V(X) \setminus V(Y)| \le N$,
and $X$ is \textit{practically disjoint} with $Y$ if $|V(X) \cap V(Y)| \le N$.

We also say that a vertex $v$
is \textit{practically universal} to a subgraph $X \subseteq G'$ if it
is adjacent to all vertices in $X$ except at most $N$ of
them. Similarly, we say a vertex $v$ is \textit{practically
independent} of a subgraph $X \subseteq G'$ if $v$ has at most $N$
neighbors in $X$.
\end{definition}

At first, the player \P uses the following algorithm for
coloring incoming vertices, which stops when it detects two
useful vertices $d_1$ and $d_2$:

\algobox{
\textbf{Algorithm~\WFD:}
For an incoming vertex $u$ sent by \D:
\begin{compactenum}
\item Let $G'_\mathrm{R}$ be the revealed part of $G'$ (i.e., colored vertices and $u$,
but not precolored vertices)

% MB: phrasing using the definition above:
\item Find a maximum clique in $G'_\mathrm{R}$ and denote it as $K_1$.
\item Find a maximum clique in $G'_\mathrm{R}$ from those which are practically disjoint with $K_1$ and denote it as $K_2$.
\item\label{wfd:if} If $|K_2|\geq S/2$ and there are two nonadjacent vertices $d_1$ and $d_2$ in $G'_\mathrm{R}$
which are both \emph{not} practically universal to $K_1$ or both \emph{not} practically universal to $K_2$:
\item\indentskip Stop the algorithm.
\item Otherwise, color $u$ using \FF. %color with lowest index
\end{compactenum}
}

While the algorithm may seem to use a huge amount of computation for
one step, we should realize that we are not concerned with time
complexity when designing the strategy for \P. In fact, even a
non-constructive proof of existence of a winning strategy would be
enough to imply existence of a PSPACE algorithm for finding it -- we
have observed already in Section \ref{sec:intro} that \onlinechrom
lies in PSPACE.  \smallskip

% MB: I feel this is already explained more or less in the Intuititon.

% Note that the algorithm \WFD{} looks for two cliques such that one should be mostly 
% from $A$ and the other mostly from $B\cup C$, supposing that a decent fraction of both has arrived.
% %i.e., $|K_1|\geq |K_2|\geq S/2$.
% However, the whole $D\cup E$ might have arrived before \WFD{} stops;
% in such a case we shall prove that colors
% used on $D$ so far are also used in $C$ or $E$, or will be used in $C$ later.

Let $v$ be the incoming vertex $u$ when \WFD{} stops; note that $v$
is not colored by the algorithm and $v$ can be from any part of $G'$.

One of the cliques $K_1$ and $K_2$ is practically a subgraph of $B \cup C$ and we denote this clique by $K_{BC}$.
The other clique must be practically a subgraph of $A$ and we denote it by $K_{A}$.
(Keep in mind that both cliques may contain up to $N$ vertices from $D\cup E$.)
We remark that some vertices from $C$ must have arrived, as $A$ and $B$ alone are
indistinguishable in Step~\ref{wfd:if} of \WFD. By the same argument, the player
\P knows whether $K_1 = K_A$ or $K_1 = K_{BC}$.

Let $d_1$ and $d_2$ be the nonadjacent vertices that caused the algorithm to stop.
We observe that $d_1,d_2 \in D$ by eliminating all other possibilities:
\begin{compactitem}
\item Neither $d_1$ nor $d_2$ is from $E$, since any vertex of $E$ is practically universal to both cliques.
\item Vertices $d_1$ and $d_2$ cannot both be from $B\cup C$, nor can both be from $A$, as they would be adjacent.
\item If $d_1$ is in $B\cup C$ and $d_2$ in $A$ (or vice versa), then we have a contradiction with the fact that $d_1$ and $d_2$ are not practically universal to the same clique.
\item If $d_1\in D$ and $d_2$ would be from $A$ or $B$ (or vice versa), then $d_1$ and $d_2$ are adjacent.
\item Finally, if $d_1\in D$ and $d_2\in C$ (or vice versa), then the clique to which they are not practically universal
cannot be the same for both, since $d_1$ is universal to the whole $A$ and $d_2$ to the whole $B\cup C$.
\end{compactitem}

% and that $v\in D$ or $v\in C$ --- %DOES NOT HOLD -- v can be from anywhere
%in the latter case, there were two nonadjacent colored vertices in $D$, but not enough vertices from $C$ so that
%both nonadjacent vertices were connected to both cliques nearly wholly.
%Observe that when \WFD{} stops, $K_2$ has size exactly $S/2$. %DOES NOT HOLD --- A, B, D arrives, then C

Having cliques $K_A$ and $K_{BC}$ and vertices $d_1,d_2 \in D$,
\P uses the following rules to recognize where an incoming or an already colored vertex $u$ belongs: 
\begin{compactitem}
\item If $u$ is practically universal to both $K_{BC}$ and $K_A$,
then $u\in E$.
\item If $u$ is practically universal to $K_{BC}$ and practically independent of $K_A$
and $u$ is adjacent to $d_1$, then $u\in B$.
\item If $u$ is practically universal to $K_{BC}$ and practically independent of $K_A$,
but there is no edge between $d_1$ and $u$, then $u\in C$.
\item If $u$ is not practically universal to $K_{BC}$, but it is practically universal to $K_A$, then $u\in A$ or $u\in D$.
\begin{compactitem}
\item Among such vertices, if there is a vertex not adjacent to $u$
or $u$ is not adjacent to a vertex in $E$
or $u$ is adjacent to a vertex in $B$, then $u\in D$;  %more than $N$ vertices from $K_{BC}$
we say that such $u$ is \textit{surely in} $D$.
\item Otherwise, the player \P cannot yet recognize whether $u\in A$ or $u\in D$.
\end{compactitem}
\end{compactitem}

The reader should take a moment to verify that indeed, the set of rules covers all possible cases
for $u$.
\medskip

Let $\tilde{A}, \tilde{B}, \tilde{C}, \tilde{D}, \tilde{E}$ be the colored parts of $G'$ when \WFD{} stops.
We observe that in the last case of the recognition the vertices from $\tilde{D}$
which are indistinguishable from $A$ form a clique; we denote it by $K_D$.
Note that all vertices in $K_D$ are connected to all vertices surely from $D$ that arrived
and $K_D$ contains all vertices in $\tilde{D}$ that are not surely in $D$.
We stress that \P does not know $K_D$ or even its size.

\boldsection{Intuition for the next step.} To start, we give a few
extremal examples of which parts of the graph can be colored when \WFD terminates:
\begin{compactenum}
\item The whole $C$ and the whole $A$ are colored, but only a clique
from $D$ is colored.
\item Some pairs of nonadjacent vertices in $D$ are colored (or even
the whole $D$) and the whole $A$ is colored, but only $S/2$ vertices
from $B\cup C$ arrived.
\item There are again some pairs of nonadjacent colored vertices in $D$
and now the whole $C$ is colored, but only $S/2$ vertices from $A$ arrived.
\end{compactenum}
Moreover, in all cases, some part of $B$ and some part of $E$ may also be
colored.

Continuing with the intuition, as \P can now recognize the parts of
the construction (with an exception of $K_D$), we would like to use
the winning strategy for \P on $G$ and \FF on the rest.  More
precisely, \P creates a virtual copy of $G$, adds vertices into it and
simulates the winning strategy on this virtual graph.

Our main problem is that some part of $D$ (namely $\tilde{D}$) is already colored. We shall prove that if $\tilde{D}$ is not a clique,
\P can ignore colors used in $\tilde{D}$ (but not the colors that it will use on $D$), since
they are already present in $C$ or $E$ or they may be used later in $C$.
If $\tilde{D}$ is a clique, it may be the case that $C$ and $A$ arrived completely
and have the same colors, thus \P cannot ignore colors used on $\tilde{D}$.

Another obstacle in the simulation is $K_D$, the hidden part of $D$.
To overcome this, \P tries to detect vertices in $K_D$
and reclassify them as surely in $D$. 
\P shall keep that all vertices in $K_D$ are connected to all currently colored vertices in $D$ and $E$,
therefore it does not matter which colors vertices in $K_D$ receive.

When \P discovers a vertex from $K_D$, it adds the vertex immediately to its simulation of $G$.
On the other hand, the size of $K_D$ increases when \D sends a vertex from $D$ which is
indistinguishable from $A$.

\smallskip

After the algorithm \WFD finishes, the player \P applies an algorithm
which simulates its winning strategy on $G$ and maintains disjoint
sets of colors $\C$ (mainly for the supernode) and $\calE$ (for $E$ and vertices surely in $D$).
Recall that \P is not able to distinguish between $A$ and
$K_D$, thus vertices in $K_D$ are treated in the same way as vertices
in $A$. More precisely, we define three color sets $\C$, $\S$ and
$\calE$ as follows:
\medskip

\begin{definition}\leavevmode
\begin{compactitem}
\item If colored vertices surely in $D$ form a clique,
let $\S$ be the set of colors that are used on a colored vertex surely in $D$ 
and that are not used in $\tilde{C}$. % or in $\tilde{E}$. % we should not color vertices in D with a color c that is in E with another color in bar G --- otherwise win. strat. on bar G may force us to use color c in D again and this would not be possible in G'
Otherwise, if there are two nonadjacent colored vertices surely in $D$,
let $\S$ be an empty set.

\item Let $\C$ be the set of colors present currently in the supernode or in $K_D$
and among vertices surely in $D$ (i.e., in $\tilde{A}, \tilde{B}, \tilde{C}$, or $\tilde{D}$)
except colors from $\S$ and colors present also in $\tilde{E}$.
\item Let $\calE$ be the set of colors assigned to vertices in $\tilde{E}$ and colors in $\S$.
\end{compactitem}
\end{definition}

We will update the color sets when we apply our algorithms and as more
vertices arrive on input; the precise updating procedure is specified
throughout the algorithms. We shall keep that $\C$ and $\calE$ are disjoint.
The winning algorithm for \P is split into two parts, initialization and coloring:

\algobox{
\textbf{Algorithm \VIRTINIT:}
\begin{compactenum}
\item \P initializes a virtual graph $\virt{G}$ by copying all vertices from $\tilde{E}$;
the colors of such vertices are inherited from $G'$.
% It shall simulate the winning strategy on $\virt{G}$ using colors from $\calE$. % MB: Just motivation, right?

\item Next, update the virtual graph with vertices surely in $D$ as follows:
\item Set an arbitrary (virtual) ordering on already arrived vertices surely in $D$.
% \item A virtual \D sends colored vertices surely in $D$ (in some order).
\item For every vertex $u$ in the ordering with a color $c$ in $G'$:
\item \indentskip If $c\in \S$, color $u$ in $\virt{G}$ using $c$.
\item \indentskip Otherwise, color $u$ in $\virt{G}$ using a color $c'$ from $\calE$ or a new color $c'$ not used in $G'$
according %TYPOTODO: Fix when margins change!
\item[] \indentskip to the winning strategy. Add $c'$ to $\calE$ if it is not there.
\end{compactenum}

% \P remembers colors of vertices in the virtual graph $\virt{G}$. % during the whole run of \WIN. % MB: Again, formally just motivation, right?
}

When adding a colored vertex $u$ that is surely in $D$ to the virtual graph $\virt{G}$, 
if $c\in \S$, coloring $u$ with $c$ in  $\virt{G}$ does not harm the winning strategy by the assumptions of Lemma~\ref{l:removePrecolVertex} --- if $\S$ is nonempty,
colored vertices surely in $D$ form a clique and \P may use any proper coloring for them in the winning strategy.
Moreover, some colors may be both in $\tilde{D}$ and in $\tilde{E}$, but the winning strategy
still works by the assumptions. %TODO is this ok? Better to check it twice. Hard case: C and A arrives completely, then clique from D and some part of E, D and E have some common colors.
%TODO this coloring is a FF coloring

Otherwise, if $c\not\in \S$, $u$ may obtain a different color in $\virt{G}$ than in the actual $G'$.

\algobox{
\textbf{Algorithm \VIRT:}
For an incoming vertex $u$ sent by \D:
\begin{compactenum}
\item For each vertex $w$ colored with $c$ in $G'$ that is surely in $D$ now,
but that was indistinguishable from vertices in $A$ before $u$ arrived
(i.e., in $K_D$ before $u$ arrived):
\begin{compactenum}
\item Add $w$ to the virtual graph $\virt{G}$. 
\item As all colored vertices in $\virt{G}$ are connected to $w$ (otherwise it would not be in $K_D$),
$w$ would obtain a new color if we add $w$ to $\virt{G}$
and use the winning strategy of \P on $\virt{G}$.
\item If $c$ is not used in $C$, remove $c$ from $\C$, add $c$ to $\calE$
and color $w$ using $c$ in $\virt{G}$.
\item Otherwise, if $c$ is present in $C$, color $w$ in $\virt{G}$ using a new color $c'$
not contained in $\C$ and $\calE$ and add $c'$ to $\calE$. (In this case, $w$ has a different
color in $G'$ and $\virt{G}$, but \P pretends that its color is $c'$ for the purpose of simulation 
and does not use $c$ for vertices in $D\cup E$ in $G'$.)
\end{compactenum}
\item If $u$ is from $C$:
\begin{compactenum}
\item If there is a color $c$ among vertices surely in $D$ such that $c\in \C$ % in $\tilde{D}$ %TODO check this
and $c$ is not present among colored vertices adjacent to $u$,
color $u$ using the smallest such color.
\item Otherwise, if there is a color in $\C$ not present among colored vertices adjacent to $u$,
color $u$ using the smallest such color.
\item Otherwise, choose a new color $c$ not contained in $\C$ and $\calE$, color $u$ with $c$, add $c$ to $\C$.
\end{compactenum}
\item If $u$ is from $A$ or $B$ or from $K_D$, i.e., indistinguishable from $A$:
\begin{compactenum}
\item If there is a color in $\C$ not present among colored vertices adjacent to $u$,
color $u$ using the smallest such color.
\item Otherwise, choose a new color $c$ not contained in $\C$ and $\calE$, color $u$ with $c$, add $c$ to $\C$.
\end{compactenum}
\item If $u$ is from $E$ or surely from $D$:
\begin{compactenum}
\item The virtual \D sends $u$ to $\virt{G}$. Simulate the winning strategy on $\virt{G}$
and color $u$ using a color $c\in \calE$ or a new color $c$ that is not in $\calE$ or $\C$;
in the latter case, add $c$ to $\calE$.
\item Color $u$ in $G'$ using the same color $c$.
\end{compactenum}
\end{compactenum}
}

We remark that the first vertex $u$ colored by the algorithm \VIRT is $v$, the vertex on which \WFD{} stops.
In a sense, if $\tilde{D}$ (the part of $D$ colored by \WFD) is not a clique,
\P pretends that in $\tilde{D}$ there are colors from the simulation,
not the colors assigned by \WFD{}.
Otherwise, if $\tilde{D}$ is a clique, \P uses colors from $G'$ in $\virt{G}$ 
which may result in renaming the colors in the winning strategy on $\virt{G}$.
%%TODO is this needed?
%By the assumptions, no color in $D$ in the virtual graph $\virt{G}$ is used in $E$,
%although this may not hold in $G'$.

Observe that any color used in $G'$ is in $\C$ or $\calE$ and that these sets of colors are disjoint.
Thus the colors used for $D\cup E$ are different
from those in the supernode except colors used in $\tilde{D}$
or colors of vertices that were in $K_D$ when they arrived
(which may happen also during the execution of \VIRT).

Note that if $u$ is from $E$ or surely from $D$, coloring with $c$ is sound; this
follows from the winning strategy on $G$ and from the fact that colors in $\calE$ are not
used in the supernode (recall that colors in $\S$ are not used in $\tilde{C}$).
%%Also coloring of $D$ is sound, since colors assigned to $D$ by \WIN{} are different from all colors used in $E$
%%by the assumption. % probably not needed

%It remains to show the following claim: % which says that for each $u\in
%\tilde{D}$ (no matter whether surely from $D$ or from $K_D$ when
%\WIN{} starts) the color of $u$ is already present somewhere in the graph $G'$. -- not quite precise
% in $\S$ or is also present in
% $\tilde{E}$, or the color is used on a vertex in $C$, no matter
% whether the vertex is in $\tilde{C}$ or not.
The key part of the analysis is captured by the following claim and its proof which
explains the design of \VIRT{}.

\begin{claim}\label{clm:ColorsFromDinC}
For any color $c\in \C$ it holds that $c$ is used on a vertex in the supernode.
\end{claim}

\begin{proof}
Any color added to $\C$ by \VIRT and not removed from $\C$ in Step 1.(c)
is used in the supernode and no color used in $E$
is in $\C$, thus it remains to show that any color $c\in \C$ used in $\tilde{D}$
(either on a vertex surely from $D$ or on a vertex in $K_D$ when \VIRT{} starts)
is present in $C$ when the whole $G'$ is colored.

We assume that $c$ is not used in $\virt{G}$ and thus also not in $\tilde{E}$,
otherwise $c$ is not in $\C$. Colors in $\tilde{D}$
and not in $\virt{G}$ (thus also not in $\S$) are in $\C$,
thus \VIRT can use $c$ only for a vertex in $C$. 
Let $u$ be the first vertex in $\tilde{D}$ that obtains color $c$ and
let $t$ be the time of coloring $u$.

Note that if $u$ is in $K_D$ at any time, then we are done, since
$u$ becomes surely in $D$ at some time (e.g., when the first vertex from $B$ arrives)
and at that time, either $c$ is already used in $C$, or
\VIRT uses $c$ also in $\virt{G}$ and removes it from $\C$.
Also, if colored vertices surely from $D$ form a clique when \VIRT starts
and $u$ is surely in $D$ at that time,
then $c$ is either used in $\tilde{C}$, or it is in $\S$ and thus not in $\C$.

Therefore we assume that $u$ is not in $K_D$ at any time, thus $u$ is surely in $D$, and 
there are two nonadjacent colored vertices surely in $D$ when \VIRT starts.
This implies $|K_{BC}| = |S|/2$ or $|K_{A}| = |S|/2$ when \VIRT{} starts which means that
$|\tilde{C}|\leq |S|/2$ or $|\tilde{A}|\leq |S|/2$.

If $|\tilde{C}|\leq |S|/2 < |S| - N$, then at least $N$ vertices from $C$ are colored 
by \VIRT{}. In this case, \VIRT{} must assign $c$ to a vertex in $C$ at some point,
since it prefers colors used in $\tilde{D}$ before
colors in $A$ or new colors and $N \geq |\tilde{D}|$.

Otherwise $|\tilde{C}| > |S|/2$, thus $|\tilde{A}|\leq |S|/2$.
If at time $t$ there is some color in $C$ not used in $A$,
let $c'$ be the smallest such color in the ordering of colors.
We observe that \WFD{} uses the color $c'$ for $u$,
since other colors allowed for coloring $u$ are some colors used in $E$ at time $t$
(which we assume that $u$ does not get), or colors in $D$ (which $u$ also does not get as it is the
first vertex in $D$ that has color $c$), or a new color not yet used anywhere,
but any new color is after $c'$ in the ordering of colors.
Hence the greedy algorithm \WFD{} prefers $c'$ and $c = c'$, thus $c$ is already used in $C$.

Otherwise, when \P colors $u$, all colors used in $C$ are also present in $A$
which also means that there are at least as many colored vertices in $A$ as in $C$ at time $t$.
%As the number of colored vertices in $A$ is at most $|S|/2$,
Let $r$ be the number of colored vertices in $C$ at time $t$.
\WFD colors more than $|S| - N - r$ vertices from $C$ after time $t$,
at most $|S|/2 - r$ of them are colored by a color already used in $A$,
thus after time $t$ \WFD colors more than $|S| - N - r - (|S|/2 - r) = |S|/2 - N$ vertices
from $C$ which do not get a color from $A$. Hence \WFD must use $c$ for a vertex in $C$ as $|S|/2 - N \geq N \geq |\tilde{D}|$.
This concludes the proof of the claim.
%OLD PROOF
%In this case we show that $c$ will be used in $C$ after time $t$.
%
%Suppose that at least $N$ vertices from $C$ are colored after $u$ arrives (by \WFD or \VIRT{}).
%Then, \WFD or \VIRT{} must assign $c$ to a vertex in $C$ at some point, since it prefers colors used in $\tilde{D}$ before
%colors in $A$ or new colors, plus we know at least $N \geq |\tilde{D}|$ vertices from $C$ will arrive after $u$.
%
%
%Assuming the opposite, more than $S - N$ vertices arrived from $C$ before time $t$,
%which means that also more than $S - N$ vertices arrived from $A$.
%Since $u$ is colored by \WFD{}, there cannot be two nonadjacent vertices
%from $D$ at time $t$, because otherwise the algorithm would stop.
%
%Observe that in this case \WFD{} stops immediately after there are two nonadjacent vertices in $D$,
%thus $\tilde{D}$ is a clique and in particular, $c$ is in $\virt{G}$ or $c$ was used in $\tilde{C}$ % or in $\tilde{E}$
%which concludes the proof.
\end{proof}

There are always $2S$ colors used in the supernode, since otherwise
there must be a color $c$ in $A$ that is not in $B\cup C$, but
both \WFD and \VIRT prefer coloring $B$ using colors in $A$ %TODO is this argument ok?
and coloring $A$ using colors from $B\cup C$. Thus $|\C| = 2S$ by Claim~\ref{clm:ColorsFromDinC}.

Since all colors in $\calE$ are used in the virtual graph $\virt{G}$ (this includes also colors in $\S$)
and \P uses the winning strategy on $\virt{G}$, we know that $|\calE|\leq k$.
Overall, \P uses at most $2S+k = k'$ colors and we proved $\chi^O(G')\leq k'$
if $\chi^O(G)\leq k$.
This concludes the proof of Lemma~\ref{l:removePrecolVertex}. %\qed
\end{proof}

\subsection{Proof of the main theorem}\label{sec:proofOfThm}

We show how to apply Lemma~\ref{l:removePrecolVertex} on the construction from Section~\ref{sec:smallPrecol}
and prove the PSPACE-completeness of \onlinechrom.

\begin{proof}[Proof of Theorem~\ref{thm:main}]
Let $\phi$ be a formula of size $n$. We first construct a graph $G_2$ with $p = \O(\log n)$ precolored vertices
$z_1, z_2, \dots z_p$ as described in Sections~\ref{sec:largePrecol} and~\ref{sec:smallPrecol}.
By Lemmas~\ref{l:smallPrecolSat} and~\ref{l:smallPrecolNotSat} 
we have $\chi^O(G_2)\leq k'$ if and only if $\phi$ is satisfiable. 
Then for each precolored vertex in $G_2$, we apply Lemma~\ref{l:removePrecolVertex} iteratively
until we obtain a graph $G_3$ with no precolored vertex such that
$\chi^O(G_3)\leq k''$ iff $\phi$ is satisfiable (for some $k''$).

The number of vertices in $G_3$ is polynomial in $n$, because $G_2$ has linearly many vertices
and the number of vertices is multiplied by a constant with each of $\O(\log n)$ applications of
Lemma~\ref{l:removePrecolVertex}. The constructions in Sections~\ref{sec:largePrecol} and~\ref{sec:smallPrecol}
and in Lemma~\ref{l:removePrecolVertex} yield a polynomial-time algorithm for computing $G'$ from $\phi$.

%We check that the graph satisfies both assumptions of Lemma~\ref{l:removePrecolVertex}
%in each iteration $i$. The first assumption,
%$\chi^O(G)\leq k$ if and only if $\phi$ is satisfiable, holds for the first iteration
%by Lemmas~\ref{l:smallPrecolSat} and~\ref{l:smallPrecolNotSat},
%and for every other iteration by the previous use
%of Lemma~\ref{l:removePrecolVertex}; note that the value of $k$ is multiplied by a constant in each iteration.

It remains to check the assumption of Lemma~\ref{l:removePrecolVertex} in each iteration $i$;
we recall it here: If $\chi^O(G)\leq k$, then
there exists a winning strategy of \P where
\P colors $E$ using \FF before two nonadjacent vertices from $D$ arrive. Moreover,
in this case if \FF assigns the same color to a vertex in $D$ and to a vertex in $E$ before two nonadjacent vertices from $D$ arrive,
\P can still color $G$ using $k$ colors.

Let $H_0$ be $G_2$ and for each iteration $i$
let $H_{i-1}$ be the graph before removing the $i$-th precolored vertex $z_i$ and
let $H_i$ be the graph created from $H_{i-1}$ by the construction in Lemma~\ref{l:removePrecolVertex}.
%in which only vertices in the subgraph $G_p$ are precolored.
%Let $z_i\in G_p$ be the precolored vertex removed in the iteration $i$
Let the nonprecolored vertices of $H_{i-1}$ be partitioned into two disjoint induced subgraphs $D_i$ and $E_i$
such that all vertices from $E_i$ are connected to $z_i$ and no vertex in $D_i$ is connected to $z_i$.
We denote by $A_i, B_i$ and $C_i$ the cliques in the supernode in $H_i$.
Thus $H_i$ consists of parts $A_i, B_i, C_i, D_i$, and $E_i$
(and possibly some precolored vertices which we do not take into account).

Note that $E_i$ contains only some nodes and $D_i$ contains the whole $G_1$
(i.e., $K_{col}$ and the gadgets for variables and clauses)
together with supernodes from previous iterations and nodes that are not in $E_i$.
%$H_{i-1}$ consists of $A_{i-1}, B_{i-1}, C_{i-1}, D_{i-1}$, and $E_{i-1}$.
We also remark that for $i > 1$ parts $A_{i-1}, B_{i-1}$, and $C_{i-1}$ of $H_{i-1}$
are in $D_i$ and that there may be some nodes
both in $D_i$ and in $E_{i-1}$ and also some nodes both in $E_i$ and in $D_{i-1}$.

We start with the first iteration.
The first part of the assumption holds for $H_0 = G_2$,
since the algorithm \GR from the winning strategy from Lemma~\ref{l:smallPrecolSat} colors the nodes
by \FF before two nonadjacent vertices from $G_1$ arrive
and $D_1$ contains $G_1$. 

For the ``moreover'' part in the first iteration,
if \FF assigns the same color $c$ to a vertex $d$ in $D_1$ and a vertex $e$ in $E_1$
before two nonadjacent vertices from $D_1$ arrive, %\P can still color $G$ using $k$ colors, 
then $e$ must be from a node that identifies $d$ and in this case \P can just pretend that
$d$ is an unrecognized vertex. In other words, the color $c$ is intended for nodes
and not used in $G_1$.

%Recall that $H_i$ consists of parts $A_i, B_i, C_i, D_i$, and $E_i$
%(except precolored vertices) and that
%$D_i$ and $E_i$ form the graph $H_{i-1}$ from the previous iteration.

For an iteration $i>1$, note that in the winning strategy on $H_{i-1}$
\P colors $H_{i-1}$ using \FF before two large and practically disjoint cliques from the supernode arrive,
one mostly from $A_{i-1}$ and the other mostly from $B_{i-1} \cup C_{i-1}$.
Since there are many pairs of nonadjacent vertices between these cliques
and these cliques are contained in $D_i$, \P can color $E_i$ using
\FF before two nonadjacent vertices from $D_i$ arrive.

For the ``moreover'' part, if \FF assigns the same color $c$ to a vertex $d$ in $D_i$ and a vertex $e$ in $E_i$
before two nonadjacent vertices from $D_i$ arrive,
then $e$ is from nodes (as $E_i$ contains only nodes)
and $d$ is either from $G_1$ as in the first iteration,
or from the supernode created in the previous iteration;
in the latter case, $e$ is in $D_j$ and $d$ is in $C_j$ for $j<i$,
since otherwise they would be connected.
(Note that $d$ cannot be from nodes, since there is a complete bipartite graph
between every two nodes and a node is never split between $E$ and $D$.)
In both cases, this does not harm the winning strategy of \P on $H_{i-1}$, %TODO write it better?
because the algorithm \VIRT assumes that \FF may use the same colors for
more vertices in $D_{i-1} \cup E_{i-1}$ or in $D_{i-1}\cup C_{i-1}$.
This concludes that the assumption of Lemma~\ref{l:removePrecolVertex}
is satisfied in each iteration.
\end{proof}

%IT MAY BE INTERESTING TO WRITE THE FOLLOWING TWO PARA. SOMEWHERE
%We remark the the algorithm \VIRT is run in each iteration, but in
%the reversed order: First, \P runs it on the graph $G_3 = H_p$ from the last 
%iteration. If \VIRT runs in an iteration $i$, then 
%there are two big colored cliques in the supernode from the iteration $i$,
%thus a lot of pairs of nonadjacent vertices for iterations $j > i$. ??? is this correct?
%
%When \P runs the algorithm \VIRT for $H_i$, it creates a virtual copy of $H_{i-1}$.
%Note that $E_i$ and only a clique from $D_i$ have colors inherited from $H_i$.
%Thus in $H_{i-1}$

%THIS IS PROBABLY NOT A GOOD IDEA:
%we observe that even a stronger assumption holds: \P can color all nodes
%and all supernodes from previous iterations using \FF before two nonadjacent vertices from $G_1$ arrive.
%The reason is that the algorithm \GR from the winning strategy from Lemma~\ref{l:smallPrecolSat}
%colors the nodes by \FF before two nonadjacent vertices from $G_1$ arrive
%and that \WFD colors the supernodes by \FF also.
%%Since $D$ contains the gadgets or $K_{col}$, the assumption holds. 

\bigskip

\textbf{Acknowledgments.} The authors thank Christian Kudahl
and their supervisor Ji\v{r}\'{i} Sgall for useful discussions on the problem.


{\small

\begin{thebibliography}{10}
\setlength\itemsep{0pt}

\bibitem{bean76}
{\sc D. R. Bean}.
{\em Effective coloration}.
The Journal of Symbolic Logic, Vol. 41, No. 2, pp. 469-480 (1976).

\bibitem{Bodlaender91}
{\sc H. Bodlaender}.
{\em On the complexity of some coloring games}.
International Journal of Foundations of Computer Science 2.02, pp. 133-147 (1991).

\bibitem{BohVel16}
{\sc M. Böhm, P. Veselý}.
{\em Online Chromatic Number is PSPACE-Complete}.
In proceedings of
{\em 27th International Workshop on Combinatorial Algorithms}
(IWOCA 2016). LNCS 9843, 16--28 (2016).

\bibitem{CMP16}
{\sc A. Csernenszky, R. R. Martin, A. Pluhár}.
{\em On the complexity of Chooser-Picker positional games}.
ArXiv preprint arXiv:1605.05430 (2016).

\bibitem{DvoVal15}
{\sc P. Dvořák, T. Valla}.
{\em On the Computational Complexity and Strategies of Online Ramsey Theory}.
In proceedings of
{\em 8th European Conference on Combinatorics, Graph Theory and Applications}
(EuroComb 2015). Electronic Notes in Discrete Mathematics 49, 729--736 (2015).

\bibitem{GryLeh88}
{\sc A. Gyárfás, J. Lehel}.
{\em First fit and on-line chromatic number of families of graphs}.
Ars Combinatoria 29C, 168--176 (1990).

\bibitem{GryKirLeh93}
{\sc A. Gyárfás, Z. Kiraly, J. Lehel}.
{\em On-line graph coloring and finite basis problems}.
Combinatorics: Paul Erdos is Eighty Volume 1., 207--214 (1993).

\bibitem{halldorsson97}
{\sc M. M. Halldórsson.}
{\em Parallel and on-line graph coloring.}
J. Algorithms, 23, 265--280 (1997).

\bibitem{halldorsson00}
{\sc M. M. Halldórsson.}
{\em Online Coloring Known Graphs.}
The Electronic Journal of Combinatorics, 7(1), R7 (2000).

\bibitem{hallszegedy94}
{\sc M. M. Halldórsson, M. Szegedy.}
{\em Lower bounds for on-line graph coloring.}
Theoretical Computer Science 130.1: 163-174 (1994).

\bibitem{kierstad98}
{\sc H. Kierstad.}
{\em On-line coloring k-colorable graphs.}
Israel Journal of Mathematics 105, 93-104 (1998).

\bibitem{kudahl13thesis}
{\sc C. Kudahl}.
{\em On-line Graph Coloring.}
Master’s thesis, University of Southern Denmark (2013).

\bibitem{kudahl14}
{\sc C. Kudahl.}
{\em Deciding the On-line Chromatic Number of a Graph with Pre-Coloring is PSPACE-Complete}.
In proceedings of
{\em 9th International Conference on Algorithms and Complexity}
(CIAC 2015). LNCS 9079, 313--324 (2015). Also arXiv:1406.1623. % reference checked

\bibitem{lst}
{\sc L. Lovász, M. Saks, W. T. Trotter.}
{\em An on-line graph coloring algorithm with sublinear performance ratio.}
Annals of Discrete Mathematics 43: 319-325 (1989).

\end{thebibliography}
}
\end{document}